\documentclass[copyright,creativecommons]{eptcs}

\pdfoutput = 1

\providecommand{\event}{EXPRESS09}
\providecommand{\volume}{??}
\providecommand{\anno}{2009}
\providecommand{\firstpage}{1}
\usepackage{amsmath}
\usepackage{amssymb}
\usepackage{latexsym}
\usepackage{epsfig}
\usepackage[ansinew]{inputenc}
\usepackage{graphicx}
\usepackage{wrapfig} 
\usepackage{breakurl}
\usepackage[final]{pdfpages}

\usepackage{ifpdf}

\newcommand{\qed}{\hspace{2mm}\nolinebreak\null\nolinebreak\hfill$\Box$}
\newtheorem{theorem}{Theorem}[section]
\newtheorem{defn}[theorem]{Definition}
\newtheorem{lemma}[theorem]{Lemma}
\newtheorem{prop}[theorem]{Proposition}

\newenvironment{proof}{\noindent {\bf Proof:}}{\bigskip}

\newcommand{\df}{=_{\text{df}}}
\newcommand{\dfi}{_{\text{df}}=}
\newcommand{\nat}{\mathbb{N}}
\newcommand{\mc}[1]{\mathcal{L}}

\newcommand{\singleton}[1]{\{{#1}\}}
\newcommand{\setof}[2]{\{{#1}\,|\,{#2}\}}

\newcommand{\twotuple}[2]{({#1},{#2})}
\newcommand{\pair}[2]{\twotuple{#1}{#2}}

\newcommand{\derives}[1]{\stackrel{#1}{\longrightarrow}}
\newcommand{\derivesL}[2]{\stackrel{#1}{\longrightarrow_#2}}

\newcommand{\TACS}{\textsc{TACS}}

\newcommand{\proc}{\text{$\mathcal{P}$}}
\newcommand{\terms}{\widehat{\proc}}

\newcommand{\act}{\text{$\mathcal{A}$}}

\newcommand{\var}{\text{$\mathcal{V}$}}

\newcommand{\nil}{\textbf{0}}
\newcommand{\res}[1]{\setminus {#1}}
\newcommand{\resset}[1]{\setminus \{{#1}\}}

\newcommand{\urgent}[1]{\mathcal{U}({#1})}

\newlength{\lrulename}
\newcommand{\Rule}[4]{\makebox[\lrulename]{{\rm #1}\hfill}
                      $\,\displaystyle\frac{#2}{#3}\,\,{#4}$}

\newcommand{\ftpo}{\raisebox{-0.7ex}{$\stackrel{\sqsupset}{\sim}$}}
\newcommand{\ftpoi}{\ftpo_{\textit{i}}}
\newcommand{\oneftpo}{\ftpo_{\textit{1}}}
\newcommand{\cftpo}{\ftpo_{\textit{c}}}
\newcommand{\twoftpo}{\ftpo_{\textit{2}}}

\newcommand{\nftpo}{\ftpo_{\textit{1-nv}}}
\newcommand{\nftpot}{\ftpo_{\textit{2-nv}}}
\newcommand{\nftpoi}{\ftpo_{\textit{i-nv}}}
\newcommand{\dftpo}{\ftpo_{\textit{1-dly}}}
\newcommand{\dftpot}{\ftpo_{\textit{2-dly}}}
\newcommand{\dftpoi}{\ftpo_{\textit{i-dly}}}

\newcommand{\uptopo}{\succeq}
\newcommand{\uptopotc}{\succeq^{+}}

\title{Robustness of a bisimulation-type faster-than preorder}
\author{Katrin Iltgen, Walter Vogler
\institute{Inst.\ f.\ Informatik,
Universit{\"a}t Augsburg\\
}
\email{\{katrin.iltgen, walter.vogler\}@informatik.uni-augsburg.de}
}
\def\titlerunning{Robustness of faster-than}
\def\authorrunning{K. Iltgen \& W. Vogler}
\begin{document}
\includepdfset{pages=-,noautoscale}
\maketitle  

\begin{abstract}
TACS is an extension of CCS where upper time bounds for delays can be specified.
Lüttgen and Vogler defined three variants of bismulation-type faster-than relations
and showed that they all three lead to the same preorder, demonstrating the
robustness of their approach. In the present paper, the operational semantics of
TACS is extended; it is shown that two of the variants still give the same
preorder as before, underlining robustness. An explanation is given why this
result fails for the third variant. It is also shown that another variant, which
mixes old and new operational semantics, can lead to smaller relations that
prove the same preorder.
\end{abstract}

\section{Introduction}
\label{sec:introduction}

To evaluate or compare the worst-case efficiency of asynchronous systems, it
is adequate to introduce upper time bounds for their actions; see
e.g.\ \cite{Lyn96} for distributed algorithms, or \cite{Vog2002} 
for Petri nets: since components can still be arbitrary fast, their 
relative speeds are indeterminate, i.e.\ one is truly dealing with 
asynchronous systems. In other words, everything that can happen when 
disregarding time can still happen in time zero; at the same time,
Zeno behaviour is not really relevant when studying worst-case efficiency.
In order to introduce upper time bounds, Milner’s CCS \cite{Mil89} is
extended by a clock prefix $\sigma$ in \cite{LV04} to obtain the 
process algebra TACS (\emph{Timed Asynchronous Communicating Systems}),
where $\sigma$ represents a potential delay of one time unit.
Since time is
treated as discrete, this means that e.g.\ process $\sigma.a.P$ can either
delay action $a$ (performing a unit time step, also denoted by $\sigma$,
leading to process $a.P$)
or skip the clock prefix and perform $a$ immediately; thus, possibly using
repeated clock prefixes, upper time bounds for communications can be specified.
An elegant type of (so-called naive) faster-than relations is introduced 
that corresponds to 
bisimulation on actions (related processes must have the same functionality)
and simulation on time steps (if the faster process performs a time step, 
i.e.\ the user has to wait, then the slower process must follow suit).

Furthermore, two variants of the faster-than relations (called delayed and 
indexed) are developed 
that are more complicated, but possibly more intuitive. As a validation 
of the first variant, it is shown that all three variants lead to the
same faster-than \emph{preorder}, demonstrating the robustness of the approach.
Furthermore, the coarsest precongruence contained in this faster-than preorder
is characterised, introducing a small modification regarding urgent actions,
and these ideas are translated to a weak setting, where $\tau$-actions are
invisible.

\medskip

In the present paper, we give new results to emphasise robustness further.
For this, we extend the operational semantics of TACS: while originally clock
prefixes could only be skipped when performing an action, we now also 
allow this when performing a time step. Upper time bounds correspond to 
the idea that there are unobservable activities on a lower level of
abstraction that can lead to varying delays; skipping a clock prefix during
a time step means that, during this time step, it becomes clear that
these activities will not lead to the maximal possible delay in this run.
For example, the process $P \equiv \sigma.\sigma.\sigma.a.\nil$ may now 
skip one $\sigma$-prefix when performing a time step and behave like 
$\sigma.a.\nil$ afterwards. This new behaviour seems more realistic; 
observe that it can become visible with a progress bar. It should be noted that, 
with this extension,
time determinism does not hold anymore; as a sort of compensation, we
gain transitivity of the $\sigma$-transition relation.

In our new setting, we look at textually the same three variants from \cite{LV04} mentioned
above, and show: two of the three variants lead to the same faster-than preorder
on TACS processes as the one in \cite{LV04}. This robustness result does not
hold for the third variant, which is an amortised bisimulation; see \cite{AKK05}
for a very similar idea. Our counter-example reveals that the idea behind
amortised bisimulation relies essentially on time determinism, so we cannot
really expect a coincidence result in this case. Robustness against small variations is often
regarded as a quality criterion, thus all in all our results further demonstrate the good quality of the TACS-approach.

We also show that analogous modifications as in \cite{LV04} lead to the 
same coarsest precongruence for our new operational semantics; the same is true for 
the weak setting of~\cite{LV04}, see~\cite{Ilt09}. But the bisimulation-like 
strong faster-than relations that serve as
witness for the faster-than precongruence are different for the old and
the new operational semantics: in some cases the first, 
in other cases the second type of relation can be smaller. As a final
contribution, we define a third type of so-called strong combined-faster-than relations
that mixes the old and
the new operational semantics and serves to demonstrate the same precongruence; 
this type of relations includes the other types, and sometimes allows a smaller relation
than any of the other two.

This paper is organised as follows.
The next section presents the syntax as well as the old and the new
operational semantics of the process algebra TACS. 
Moreover, we get familiar with the nature of the new $\sigma$-transitions and prove
their transitivity.
In Section 3, we compare the original time steps to those newly introduced 
here; the key concept is a syntactic faster-than relation from~\cite{LV04}. Subsequently, we demonstrate the robustness of the naive 
faster-than preorder against the transition extension in Section 4.
Section 5 proves the robustness of the delayed faster-than preorder, while
Section 6 demonstrates the defect of the extended indexed faster-than preorder.
In Section 7, we establish the robustness of the precongruence; we further introduce 
the strong combined-faster-than relations and give two example processes where 
such a relation can be much smaller than a relation of one of the other two types.
Finally, we draw a short conclusion in Section 8.
The preliminary version~\cite{Ilt09} of this paper contains a number of proof details
that are omitted here.


\section{TACS}
\label{sec:TACS} 

In this section, we introduce TACS~\cite{LV04} as an extension of CCS by the
clock prefix $\sigma$, representing a delay of up to one unit of time.
We define the operational semantics of TACS, including our new extension,
and we give some first results.


Let $\Lambda$ be a countable set of action names or ports $a,b,c$;
$\overline{\Lambda} \df\{\overline{a} | a \in \Lambda\}$ is the set of 
complementary action names $\overline{a},\overline{b},\overline{c}$, and
$\act \df \Lambda \cup \overline{\Lambda} \cup \{\tau\}$ is the set of 
all actions $\alpha,\beta,\gamma$, including the internal action $\tau$. 
As usual, $\overline{\overline{a}}\df a$ for all $a \in \Lambda$,
and an action~$a$ will communicate with its
complement~$\overline{a}$ to produce the internal action~$\tau$.
\medskip

A TACS term is defined as follows, where the operators have the usual
meaning:
\begin{equation*}
P \;\;::=\;\; \nil      \;\;|\;\;
                x           \;\;|\;\;
                \alpha.P  \;\;|\;\;
                \sigma.P  \;\;|\;\;
                P+P         \;\;|\;\;
                P|P         \;\;|\;\;
                P\res{L}   \;\;|\;\;
                P[f]        \;\;|\;\;
                \mu x. P 
\end{equation*}
where~$x$ is a \emph{variable} taken from a countably infinite
set~$\var$ of variables, $L \subseteq \act \setminus \singleton{\tau}$
is a finite \emph{restriction set}, and $f: \act \rightarrow \act$ is a
\emph{finite relabelling}.  A finite relabelling satisfies the properties
$f(\tau) = \tau$, $f(\overline{a}) = \overline{f(a)}$, and $|
\setof{\alpha}{f(\alpha) \not= \alpha} | < \infty$.  The set of all
terms is abbreviated by~$\terms$ and, for convenience, we define
$\overline{L} \df \setof{\overline{a}}{a \in L}$.  We use the standard
definitions for \emph{free} and \emph{bound} variables (where~$\mu x$
binds~$x$), and \emph{open} and \emph{closed} terms.  
$P[Q/x]$ stands for the term
that results when substituting every free occurrence of~$x$ in~$P$
by~$Q$.  A variable is called \emph{guarded} in a term if each
occurrence of the variable is in the scope of an action prefix.  We
require for terms of the form~$\mu x.  P$ that~$x$ is guarded in~$P$.
Closed, guarded terms are referred to as \emph{processes}, with the set
of all processes written as~$\proc$, and syntactic equality is denoted
by~$\equiv$.
\medskip


\begin{table}[htb]
\caption{Urgent action sets}
\begin{center}
$\begin{array}{@{}l@{\,}c@{\,}l@{\quad\!}
                    l@{\,}c@{\,}l@{\quad\!}l@{\,}c@{\,}l@{}}
\hline
\urgent{\sigma.P} & \df & \emptyset
&
\urgent{\nil} = \urgent{x} & \df & \emptyset
&
\urgent{P\res{L}} & \df & \urgent{P} \setminus (L\cup\overline{L})
\\
\urgent{\alpha.P} & \df & \singleton{\alpha}
&
\urgent{P+Q} & \df & \urgent{P} \cup \urgent{Q}
&
\urgent{P[f]} & \df & \setof{f(\alpha)}{\alpha \in \urgent{P}} 
\\
\urgent{\mu x. P} & \df & \urgent{P}
&
\urgent{P|Q} & \df & 
\multicolumn{4}{l}{%
  \urgent{P} \cup \urgent{Q} \cup 
  \setof{\tau}{\urgent{P} \cap \overline{\urgent{Q}} \not= \emptyset}
}
\\
\hline
\end{array}$
\end{center}
\label{table:urgent}
\end{table}
\begin{table}[htb]
\caption{Operational semantics for \TACS\ (action transitions)}
\begin{center}
\begin{tabular}{@{}l@{\qquad\;\;}l@{\qquad\;\;}l@{}}
\hline
\\
\Rule{Act}
     {-\!\!\!-}
     {\alpha.P \derives{\alpha} P}
     {}
&
\Rule{Pre}
     {P \derives{\alpha} P'}
     {\sigma.P \derives{\alpha} P'}
     {}
&
\Rule{Rec}
     {P \derives{\alpha} P'}
     {\mu x. P \derives{\alpha} P'[\mu x. P /x]}
     {}
\\ \
\\
\Rule{Sum1}
     {P \derives{\alpha} P'}
     {P+Q \derives{\alpha} P'}
     {}
&
\Rule{Com1}
     {P \derives{\alpha} P'}
     {P|Q \derives{\alpha} P'|Q}
     {}
&
\Rule{Com3}
     {P \derives{a} P' \quad
      Q \derives{\overline{a}} Q'}
     {P|Q \derives{\tau} P'|Q'}
     {}
\\ \
\\
\Rule{Rel}
     {P \derives{\alpha} P'}
     {P[f] \derives{f(\alpha)} P'[f]}
     {}
&
\multicolumn{2}{l}{%
\Rule{Res}
     {P \derives{\alpha} P'}
     {P\res{L} \derives{\alpha} P'\res{L}}
     {\alpha \notin L\cup\overline{L}}
}
\\ \
\\
\hline
\end{tabular}
\end{center}
\label{table:act-opsem}
\end{table}

As a basis for the operational semantics, we first define the set $\urgent{P}$ 
of \emph{urgent} actions of $P$, i.e.\ those initial actions that cannot be
delayed because of a $\sigma$-prefix, see Rule (Pre). On the
basis of this set, we define the action transitions with the 
SOS-rules in Table~\ref{table:act-opsem}; there are symmetric rules
(Sum2) and (Com2) for (Sum1) and (Com1). These are not influenced by
our extension, and are standard except for Rule (Pre), explained
in the introduction.

The time steps, i.e.\ the $\sigma$-transitions, are defined in
Table~\ref{table:clock-opsem}. We write $\derivesL{\sigma}1$ for the
type-1 time steps according to~\cite{LV04} -- they do not use Rule (tNew) --
and $\derivesL{\sigma}2$ for our extended setting, using all rules.
Observe that $a.P$ can perform a time step, since it might have to
wait for a synchronisation partner that can delay $\overline{a}$ or is
unable to perform it. Furthermore, special care is taken in Rule (tCom)
such that all in all, a time step is possible iff there is no urgent
$\tau$ (\emph{Maximal Progress Assumption}). E.g. 
$a.\nil \;|\;\overline{a}.\nil$ cannot perform a time step, while
$a.\nil \;|\; \sigma.\overline{a}.\nil$ can.

We will write $\derivesL{\sigma}i^{+}$ (and $\derivesL{\sigma}i^{\ast}$)
for the transitive (and the reflexive-transitive resp.)
closure of $\derivesL{\sigma}i$ for $i \in \{1,2\}$, and similarly for
other relations.

\begin{table}[h]
\caption{Operational semantics for \TACS\ (clock transitions)}
\begin{center}
\begin{tabular}{@{}l@{\qquad\;\,}l@{\qquad\;\,}l@{}}
\hline
\\
\Rule{tNil}
     {-\!\!\!-}
     {\nil \derivesL{\sigma}i \nil}
     {}
&
\Rule{tRec}
     {P \derivesL{\sigma}i P'}
     {\mu x. P \derivesL{\sigma}i P'[\mu x. P /x]}
     {}
&
\Rule{tRes}
     {P \derivesL{\sigma}i P'}
     {P\res{L} \derivesL{\sigma}i P'\res{L}}
     {} 
\\ \
\\
\Rule{tAct}
     {-\!\!\!-}
     {a.P \derivesL{\sigma}i a.P}
     {}
&
\Rule{tSum}
     {P \derivesL{\sigma}i P' \quad Q \derivesL{\sigma}i Q'}
     {P+Q \derivesL{\sigma}i P'+Q'}
     {}
&
\Rule{tRel}
     {P \derivesL{\sigma}i P'}
     {P[f] \derivesL{\sigma}i P'[f]}
     {}
\\ \
\\
\Rule{tPre}
     {-\!\!\!-}
     {\sigma.P \derivesL{\sigma}i P}
     {}
&
\multicolumn{2}{l}{%
\Rule{tCom}
     {P \derivesL{\sigma}i P' \quad Q \derivesL{\sigma}i Q'}
     {P|Q \derivesL{\sigma}i P'|Q'}
     {\tau \notin \urgent{P|Q}}
}
\\ \
\\
\Rule{tNew}
     {P \derivesL{\sigma}2 P'}
     {\sigma.P \derivesL{\sigma}2 P'}
     {}
\\ \
\\
\hline
\end{tabular}
\end{center}
\label{table:clock-opsem}
\end{table}

Clearly, type-1 time steps like 
$\sigma.\sigma.\sigma.a.\nil \derivesL{\sigma}1 \sigma.\sigma.a.\nil$
are also of type 2, but there are additional
type-2 time steps as e.g.\ 
$\sigma.\sigma.\sigma.a.\nil \derivesL{\sigma}2 \sigma.a.\nil$, which
corresponds to a sequence of two type-1 time steps. But things are not
always that easy: e.g.\ 
$\sigma.\sigma.\sigma.a.\nil\;|\;\sigma.\sigma.a.\nil\derivesL{\sigma}2
a.\nil\;|\;\sigma.a.\nil$ reaches a process that cannot be reached with 
(several) type-1 steps. So the new behaviour we model here is a significant 
extension, and its treatment needs non-trivial proofs.

%
%


Before proceeding, it is convenient to establish a lemma that highlights the interplay between our transition relation and substitution and
that will be employed in some of the following proofs. It is essentially
taken from~\cite{LV04} and also holds in the extended setting.

\begin{lemma}
  Let $P, P', Q \in \terms$ and $\gamma \in \act \cup \{\sigma\}$.
  \begin{enumerate}
  \item\label{part:1-rec} $P \derivesL{\gamma}2 P'$ implies $P[\mu y. Q /
    y] \derivesL{\gamma}2 P'[\mu y. Q / y]$.
  \item\label{part:2-rec} $y$ guarded in~$P$ and
    $P[\mu y. Q / y] \derivesL{\gamma}2 P'$ implies \\ \hspace*{4.0em}
    $\exists P'' \in \terms.\, P \derivesL{\gamma}2
    P''$ and $P' \equiv P''[\mu y. Q / y]$.
  \end{enumerate}
\label{lem:rec}
\end{lemma}

Now we are able to prove our first result that $\derivesL{\sigma}2$ is 
transitive. 
In its proof and also in the future, we will use 
the following preservation of guardedness under a time step (which is
not hard to see): 
$P \derivesL{\sigma}i P'$ for some $i \in \{1,2\}$ and $x$ guarded 
in $P$ implies that $x$ is also guarded in $P'$.

\begin{prop} Let $P, P', P'' \in \terms$.\\
$P \derivesL{\sigma}2 P' \derivesL{\sigma}2 P''$ implies $P \derivesL{\sigma}2 P''$.
\label{prop:transitive-type-2}
\end{prop}

\begin{proof}
This proposition can be proved by induction on the structure of $P$. The two more
interesting cases are:

\begin{enumerate}
\item[(1)] Let $P \equiv \sigma.P_{1}$. 
If $\sigma.P_{1} \derivesL{\sigma}2 P_{1}$ by (tPre) and 
$P_{1} \derivesL{\sigma}2 P''$, we can infer 
$\sigma.P_{1} \derivesL{\sigma}2 P''$ by (tNew).\\
If $\sigma.P_{1} \derivesL{\sigma}2 P'$ due to 
$P_{1} \derivesL{\sigma}2 P'$ by (tNew) and
$P' \derivesL{\sigma}2 P''$, then we have
$P_{1} \derivesL{\sigma}2 P''$ by induction; using (tNew) we conclude
$\sigma.P_{1} \derivesL{\sigma}2 P''$.

\item[(2)] Let $P \equiv \mu x.P_{1}$. Consider the case $\mu x.P_{1} \derivesL{\sigma}2 P_{1}'[\mu x.P_{1} /x]$ due to 
$P_{1} \derivesL{\sigma}2 P_{1}'$ by (tRec) and 
$P_{1}'[\mu x.P_{1} /x] \derivesL{\sigma}2 P''$. By our requirements for 
$\mu x.P_{1}$, $x$ is guarded in $P_{1}$ and, hence, also in $P_{1}'$.
Using Lemma \ref{lem:rec}(\ref{part:2-rec}), we obtain 
$P_{1}' \derivesL{\sigma}2 P_{1}''$ and 
$P'' \equiv P_{1}''[\mu x.P_{1} /x]$ for some $P_{1}'' \in \proc$, since 
$x$ is guarded in $P_{1}'$. By induction, we infer 
$P_{1} \derivesL{\sigma}2 P_{1}''$ from 
$P_{1} \derivesL{\sigma}2 P_{1}' \derivesL{\sigma}2 P_{1}''$, and conclude 
$\mu x.P_{1} \derivesL{\sigma}2 P_{1}''[\mu x.P_{1} /x] \equiv P''$ by (tRec). \qed
                                       
\end{enumerate}
\end{proof}

\section{Relating $\derivesL{\sigma}2$ and $\derivesL{\sigma}1$}
\label{subsec:coherence}

The key to describing how $\derivesL{\sigma}2$ can be matched by
$\derivesL{\sigma}1$ and to the new robustness results lies in the syntactic 
faster-than relation $\uptopo$ from~\cite{LV04}
and its transitive closure.

\begin{defn}\rm
  The relation $\uptopo \;\subseteq \terms \times \terms$ is defined as the
  smallest relation satisfying the following properties, for all
  $P, P', Q, Q' \in \terms$.\\

  $
  \begin{array}{@{}r@{\;\;\;\,}l@{\;\,}l@{\,}l@{\;\;\,}l@{}}
  \text{Always:}
  &
  (1) & P \uptopo P
  &
  (2) & P \uptopo \sigma.P
  \\
  \text{If $P' \uptopo P$, $Q' \uptopo Q$:}
  &
  (3) & P'|Q' \uptopo P|Q
  &
  (4) & P'+Q' \uptopo P+Q
  \\
  &
  (5) & P'\res{L} \uptopo P\res{L}
  &
  (6) & P'[f] \uptopo P[f]
  \\
  \text{If $P' \uptopo P$, $x$~guarded in~$P$:}
  &
  (7) & P'[\mu x.\, P / x] \uptopo \mu x.\, P   
  &
      & $\qed$
  \end{array}
  $
\label{def:uptopo}
\end{defn}

Observe that the syntactic relation is defined for arbitrary open terms. We note
some technical results, partly taken from \cite{LV04}:
  
\begin{lemma}
Let $P, P', Q \in \terms$ such that $P' \uptopo P$, and let $y \in \var$.
  \begin{enumerate}
    \item\label{part:1-upto-guarded}\emph{\cite[Lemma 7(1)]{LV04}} Then $y$ is guarded in~$P$ if
      and only if $y$ is guarded in~$P'$. 
    \item\label{part:1-upto-aux1}\emph{\cite[Lemma 7(2)]{LV04}} $P'[Q/y] \uptopo P[Q/y]$.
    \item\label{part:1-upto-urgent}\emph{\cite[Lemma 8(2)]{LV04}} $\urgent{P'} \supseteq \urgent P$.
  \end{enumerate}
\label{lem:upto-aux1-rec}
\end{lemma}

The following lemma establishes several properties of $\uptopotc$ 
similar to those of the syntactic relation in Def.
\ref{def:uptopo}.

\begin{lemma}
  Let $P, P', Q, Q' \in \terms$.
  \[
  \begin{array}{@{}l@{\;\;\;\,}l@{\;\,}l@{\;\,}}
  \text{If $P' \uptopotc P$, $Q' \uptopotc Q$ then:}
  &
  (1) & P'|Q' \uptopotc P|Q \\
  &
  (2) & P'+Q' \uptopotc P+Q \\
  \text{If $P' \uptopotc P$ then:}
  &
  (3) & P' \res{L} \uptopotc P \res{L}  \\
  &
  (4) & P'[f] \uptopotc P[f]  \\
  \text{If $P' \uptopotc P$, $x$ guarded in $P$ then:}
  &
  (5) & P'[\mu x.\, P / x] \uptopotc \mu x.\, P   \\
  \end{array}
  \]
\label{lem:uptopotc}
\end{lemma}

\begin{proof}
In the proof of (1) and (2), one has to deal with the case that 
$P' \uptopotc P$ and $Q' \uptopotc Q$ hold because of $\uptopo$-chains of
different length: one simply extends the shorter chain using
 Def.~\ref{def:uptopo}.(1). Parts (3) and (4) are easier.

For (5), take $n \ge 1$ and $P_{0},...,P_{n} \in \terms$ with
$P' \equiv P_{0} \uptopo P_{1} \uptopo \cdots \uptopo P_{n} \equiv P$. Since $x$ is 
guarded in $P$, we infer $P_{n-1}[\mu x.\, P_{n} / x] \uptopo \mu x.\, P_{n}$ 
from $P_{n-1} \uptopo P_{n}$ by using Def. \ref{def:uptopo}(7).
Further, we obtain $P_{i-1}[\mu x.\, P_{n} / x] \uptopo P_{i}[\mu x.\, P_{n} / x]$, due to $P_{i-1} \uptopo P_{i}$  
for $1 \le i \le n-1$, by Lemma \ref{lem:upto-aux1-rec}(\ref{part:1-upto-aux1})
and are done.
\end{proof}

With a time step, a process should turn into a process that should be faster 
with everything it does; we can now prove this in terms of $\uptopo$ -- the main
achievement here lies in the treatment of recursion. The other important point about 
$\uptopo$ is that it describes a faster-than relationship in the sense of the semantic definitions in the next section; this holds by Proposition~\ref{prop:up-to-naives} 
for one variant, and Theorems~\ref{thm:1-naive-2-naive-coincidence} 
and~\ref{thm:1-delayed-2-delayed-coincidence} transfer this to other variants.

\begin{lemma}
  Let $P, P' \in \terms$.
  \begin{enumerate}
   \item\label{part:1-up-to-sigma}\emph{\cite[Prop. 9(1)]{LV04}} $P \derivesL{\sigma} 1 P'$ implies $P'
   \uptopo P$, for all terms $P, P' \in \terms$.
   \item\label{part:2-up-to-sigma} $P \derivesL{\sigma} 2 P'$ implies $P'
   \uptopotc P$, for all terms $P, P' \in \terms$.
   \end{enumerate}
  \label{lem:up-to-sigma}
\end{lemma}

\begin{proof}
We prove Part~(\ref{part:2-up-to-sigma}) by induction on the inference of 
$P \derivesL{\sigma}2 P'$.
  \begin{enumerate}
   \item[tNil] $P \equiv \nil \equiv P'$. Since $\uptopo$ $\subseteq$ $\uptopotc$, $\nil \uptopotc \nil$ holds by using Def.~\ref{def:uptopo}(1).
   \item[tAct] $P \equiv a.P'' \equiv P'$. Since $\uptopo$ $\subseteq$ $\uptopotc$, $a.P'' \uptopotc a.P''$ holds by using Def.~\ref{def:uptopo}(1).
   \item[tPre] $P \equiv \sigma.P'$. Since $\uptopo$ $\subseteq$ $\uptopotc$, $P' \uptopotc \sigma.P'$ holds by using Def.~\ref{def:uptopo}(2).
   \item[tNew] $P \equiv \sigma.P_{1}$. 
   Let $\sigma.P_{1} \derivesL{\sigma}2 P'$ due to $P_{1} \derivesL{\sigma}2 P'$.
The latter implies $P' \uptopotc P_{1}$ by induction, and with  
$P_{1} \uptopo \sigma.P_{1}$ by Def.~\ref{def:uptopo}(2),
we conclude $P' \uptopotc  \sigma.P_{1}$.
    \item[tRec] $P \equiv \mu x.\, P_{1}$ and $P' \equiv P_{2}[\mu x.\, P_{1} / x]$.\\ 
Let $\mu x.\, P_{1} \derivesL{\sigma}2 P_{2}[\mu x.\, P_{1} / x]$ due to 
$P_{1} \derivesL{\sigma}2 P_{2}$. By induction, the latter implies 
$P_{2} \uptopotc P_{1}$. Since $x$ is guarded in $P_{1}$, we can infer 
$P_{2}[\mu x.\, P_{1} / x] \uptopotc \mu x.\, P_{1}$ from $P_{2} \uptopotc P_{1}$ by Lemma~\ref{lem:uptopotc}(5) . 
   \item[tSum] $P \equiv P_{1}+Q_{1}$ and $P' \equiv P_{2}+Q_{2}$.\\ 
                                   Since $P \derivesL{\sigma}2 P'$, we have $P_{1} \derivesL{\sigma}2 P_{2}$ and ${Q}_{1} \derivesL{\sigma}2 Q_{2}$.
                                   $P_{2} \uptopotc P_{1}$ and $Q_{2} \uptopotc Q_{1}$ follows by induction hypothesis and 
                           $P_{2}+Q_{2} \uptopotc P_{1}+Q_{1}$ results by application of Lemma~\ref{lem:uptopotc}(2).
   \item[tCom] The treatment of this case is analogous to case tSum and uses Lemma~\ref{lem:uptopotc}(1).
   \item[tRes] $P \equiv P_{1} \res{L}$ and $P' \equiv P_{2} \res{L}$.\\ We obtain
                                   $P_{2} \uptopotc P_{1}$ by induction and 
                           $P_{2} \res{L} \uptopotc P_{1} \res{L}$ by Lemma~\ref{lem:uptopotc}(3).
   \item[tRel] This case follows in analogy to case tRes, using Lemma~\ref{lem:uptopotc}(4). \qed
  
   \end{enumerate}

\end{proof}

Now we will relate the result of a type-2 time step with the only result of
a type-1 time step. Consider e.g.\
$P \df \sigma.\sigma.\sigma.a.\nil \,|\,\sigma. \overline{a}.\nil \,|\, 
\sigma.a.\nil$
$\derivesL{\sigma} 2$ $a.\nil \,|\,\overline{a}.\nil \,|\, a.\nil\,\dfi  P'$, while 
the only enabled type-1 time step leads to
$ \sigma.\sigma.a.\nil \,|\,\overline{a}.\nil \,|\, a.\nil\,\dfi  P'' $; observe 
$P'' \not\!\!\derives{\sigma}$ as $\tau \in \urgent {P''}$.
We see that $P'$ results from $P''$ by removing some leading $\sigma$-prefixes,
so they should be related by $\uptopotc$.

\begin{prop} 
  Let $P, P' \in \terms$.
  $P \derivesL{\sigma} 2 P'$ implies $\exists P'' \in \terms$. $P \derivesL{\sigma} 1 P''$ and $P' \uptopotc P''$.
\label{lem:coherence}
\end{prop}

\begin{proof}
The proof is an induction on the inference of $P \derivesL{\sigma} 2 P'$. Most 
cases are straightforward, often using Lemma~\ref{lem:uptopotc}; we only show two:
\begin{enumerate}

\item[tNew] $P \equiv \sigma.P''$. Observe that 
$\sigma.P'' \derivesL{\sigma}1 P''$, and let $\sigma.P'' \derivesL{\sigma}2 P'$ 
due to $P'' \derivesL{\sigma}2 P'$ by (tNew). The latter implies
$P' \uptopotc P''$  by Lemma \ref{lem:up-to-sigma}(\ref{part:2-up-to-sigma}).

\item[tRec] $P \equiv \mu x.\, P_{1}$. Let 
$\mu x.\, P_{1} \derivesL{\sigma}2
P_{1}'[\mu x.\, P_{1} / x]$ due to $P_{1} \derivesL{\sigma}2 P_{1}'$ by (tRec$_{2}$). 
\\By the induction hypothesis, there exists $P_{2}'$ such that 
$P_{1} \derivesL{\sigma} 1 P_{2}'$ and $P_{1}' \uptopotc P_{2}'$.\\ Thus, we 
may infer $\mu x.\, P_{1} \derivesL{\sigma}1 P_{2}'[\mu x.\, P_{1} / x]$ by 
(tRec). Applying Lemma \ref{lem:upto-aux1-rec}(\ref{part:1-upto-aux1})
to the $\uptopo$-chain behind $P_{1}' \uptopotc P_{2}'$, we obtain
$ P_{1}'[\mu x.\, P_{1} / x] \uptopotc P_{2}'[\mu x.\, P_{1} / x]$. \qed  
\end{enumerate}
\end{proof}

Generalising a result from \cite{Ilt09}, but with the same proof, one can show
for any process $P$, where any occurrence of parallel composition or choice is
\lq guarded\rq\, that $P \derivesL{\sigma}2 P'$ implies 
$P \derivesL{\sigma}1^{+} P'$.


\section{The naive faster-than preorders}
\label{sec:faster-than}

In \cite{LV04}, the naive faster-than preorder is introduced as an elegant and 
concise candidate for a faster-than preorder: the faster and the slower process are simply linked by a 
relation that is a simulation for time steps and a strong bisimulation for actions.
The definition of the 1-naive faster-than preorder is adopted from \cite{LV04}
and extended to a second variant by considering our new type-2 transitions:

\begin{defn}\rm
  For $i \in \{1,2\}$, a relation $\mathcal{R}\subseteq \proc\times\proc$ is an 
\emph{i-naive
    faster-than relation} if the following conditions hold for all
  $\pair{P}{Q}\in \mathcal{R}$ and $\alpha \in \act$.

  \begin{enumerate}
  \item\label{part:1-1-naive-preorder}  
        $P \derives{\alpha} P'$ implies 
        $\exists Q'.\, Q \derives{\alpha} Q'$
        and $\pair{P'}{Q'}\in \mathcal{R}$.
  \item\label{part:2-1-naive-preorder} 
        $Q \derives{\alpha} Q'$ implies 
        $\exists P'.\, P \derives{\alpha} P'$
        and $\pair{P'}{Q'}\in \mathcal{R}$.
  \item\label{part:3-1-naive-preorder} 
        $P \derivesL{\sigma} i P'$ implies 
        $\exists Q'.\, Q \derivesL{\sigma} i Q'$
        and  
        $\pair{P'}{Q'}\in \mathcal{R}$.
  \end{enumerate}
  We write $P \,\nftpoi\, Q$ if $\pair{P}{Q} \in \mathcal{R}$ for some i-naive
  faster-than relation~$\mathcal{R}$, and call $\nftpoi$ 
\it{i-naive faster-than preorder}.\qed
\label{def:1-naive-preorder}
\end{defn}

Part (1) and (2) require that the faster and the slower process are functionally 
equivalent, in other words, $\mathcal{R}$ is a strong bisimulation. The i-naive faster-than relation refines strong bisimulation 
since additionally any time step of the faster process (corresponding
to a delay of its user) must be 
simulated by the slower one. Extra time steps $Q$ might perform are not 
considered (cf.\ Sec.~\ref{subsec:delayed-faster-than-preorders}), 
intuitively because they do not change the functional behaviour 
of $Q$. (Formally, $Q \derivesL{\sigma} 2  Q'$ implies $Q' \uptopotc Q$ by Lemma \ref{lem:up-to-sigma}(\ref{part:2-up-to-sigma}), thus $Q'$ and $Q$ are strongly bisimilar
by Proposition \ref{prop:up-to-naives}(\ref{part:2-up-to-naives}) below.) Note that faster-than holds for equally fast processes, i.e.\ it is
not strict.

As usual, one can show that $\nftpoi$ is an i-naive faster-than relation and a
preorder; in particular, the composition of two such relations is again one.
The i-naive faster-than preorder is defined on processes; it can be extended to
open terms as usual by considering all possible substitutions with processes.


In the sequel, we will prove that the faster-than preorders $\,\nftpo$ and 
$\,\nftpot$ coincide. For this, we take from \cite{LV04} that the 
syntactic $\uptopo$ satisfies the definition of a 1-naive faster-than relation
and show the same for $\uptopotc$. The same results for the 2-naive case can be 
useful, but we are able to avoid their tedious proofs.

\begin{prop}
\begin{enumerate}
\item\label{part:1-up-to-naives} The relation~$\uptopo$ satisfies the
  defining clauses of a 1-naive faster-than relation, also on open terms;
  hence, $\uptopo$ restricted to processes is a 1-naive faster-than
  relation and
  $\uptopo_{| \proc \times \proc} \;\,\df\,\;
  \uptopo \cap\, (\proc \times \proc)
  \;\subseteq \nftpo$.
\item\label{part:2-up-to-naives} The same holds for~$\uptopotc$.
\end{enumerate}
\label{prop:up-to-naives}
\end{prop}

\begin{proof} To obtain Part~(\ref{part:2-up-to-naives}) from 
Part~(\ref{part:1-up-to-naives}), observe that $\uptopotc$ equals 
$\uptopo^{0} \cup \uptopo^{1} \cup \dots$ and the closure of 
1-naive faster-than relations under union and composition.\qed
\end{proof}
  \begin{theorem}[Coincidence~I]
  The preorders~$\nftpo$ and~$\nftpot$ coincide.
\label{thm:1-naive-2-naive-coincidence}
\end{theorem}

\begin{proof} There are two places, where we insert some text in square
brackets. With this text, this is essentially the proof of 
Theorem~\ref{thm:1-2-strong-faster-coincidence}; here, the text should be
ignored.

To prove $\nftpo \subseteq \nftpot$, we show that $\nftpo$ satisfies the 
definition of a 2-naive faster-than relation; clearly, we have only to 
consider Part~(3).
  Hence, consider some arbitrary processes $P$ and $Q$ such that 
$P \,\nftpo\, Q$.
        If $P \derivesL{\sigma}2 P'$ for some process $P'$, then 
$P \derivesL{\sigma}1 P''$ for some process $P''$ satisfying $P' \uptopotc P''$
  by Proposition~\ref{lem:coherence}.
  By definition of $\nftpo$, [$\urgent{Q} \subseteq \urgent{P}$ and]
there exists some $Q''$ with 
  $Q \derivesL{\sigma}1 Q''$ and $P'' \,\nftpo\, Q''$; hence, also
  $Q \derivesL{\sigma}2 Q''$. 
  $P' \uptopotc P''$ implies $P' \,\nftpo \, P''$ by 
  Proposition~\ref{prop:up-to-naives}(\ref{part:2-up-to-naives}), and
  we are done by transitivity of $\nftpo$.
 
  For the inverse inclusion $\nftpot \subseteq \nftpo$, consider $P$ and $Q$ 
such that $P \,\nftpot\, Q$.
 If $P \derivesL{\sigma}1 P'$ for some process $P'$, then
  $P \derivesL{\sigma}2 P'$ 
by definition of  $\nftpot$, $Q \derivesL{\sigma}2 Q'$ for some $Q'$ such that 
[$\urgent{Q} \subseteq \urgent{P}$ and] $P' \,\nftpot\, Q'$ .
  By Proposition~\ref{lem:coherence}, there exists some $Q''$ with
  $Q \derivesL{\sigma}1 Q''$ and $Q' \uptopotc Q''$. Hence, $Q' \,\nftpot\, Q''$
follows from $\uptopotc \subseteq \nftpo \subseteq \nftpot$, and
  we are done by transitivity of $\nftpot$.
\qed
\end{proof}


\section{The delayed-faster-than-preorders}
\label{subsec:delayed-faster-than-preorders}

In the course of design choices in \cite{LV04}, an alternative candidate for a faster-than relation on processes is the delayed faster-than preorder.
Since the slower process might take more time, this preorder allows the
slower process additional time steps when matching an action or time step.
We define the i-delayed faster-than preorder, where case $i=1$
is adopted from \cite{LV04}. 

Before giving the definition, we present a technical lemma, highlighting an
important property of the transitive closure of the syntactic relation.
The result is intuitively convincing because, if the faster process 
skips time steps by performing a 'real' type-2 time step, it can only become 
even faster than the slower process $Q$.

\begin{lemma}
  Let $P, P' \in \terms$  such that $P \uptopotc P'$.
  Then $P \derivesL{\sigma}2 P_{1}$ implies 
$\exists P_{1}'.P' \derivesL{\sigma}1 P_{1}'$ and $P_{1} \uptopotc P_{1}'$.   
\label{lem:uptopotc-sigma1}
\end{lemma}

\begin{proof}
$P \derivesL{\sigma}2 P_{1}$ implies $P \derivesL{\sigma}1 P_{1}''$ and
$P_{1} \uptopotc P_{1}''$  by using Proposition~\ref{lem:coherence}. Then we
get $P' \derivesL{\sigma}1 P_{1}'$ and $P_{1}'' \uptopotc P_{1}'$ by 
Proposition~\ref{prop:up-to-naives}(\ref{part:2-up-to-naives}) and are
done by the obvious transitivity of $\uptopotc$.
\qed
\end{proof}

\begin{defn}\rm
 For $i \in \{1,2\}$, a relation $\mathcal{R}\subseteq \proc\times\proc$
 is an \emph{i-delayed faster-than relation} if the following conditions hold for all
  $\pair{P}{Q}\in \mathcal{R}$ and $\alpha \in \act$.
  \begin{enumerate}
  \item $P \derives{\alpha} P'$ implies 
        $\exists Q'.\, Q \derivesL{\sigma}i^{\ast}\!\derives{\alpha}\;
                         \derivesL{\sigma}i^{\ast} Q'$
        and $\pair{P'}{Q'}\in \mathcal{R}$.
  \item $Q \derives{\alpha} Q'$ implies
        $\exists P'.\, P \derives{\alpha} P'$
        and $\pair{P'}{Q'}\in \mathcal{R}$.
  \item $P \derivesL{\sigma}i P'$ implies 
        $\exists Q'.\, Q \derivesL{\sigma}i^{+} Q'$
        and  
        $\pair{P'}{Q'}\in \mathcal{R}$.
  \end{enumerate}
  We write $P \,\dftpoi Q$ if $\pair{P}{Q} \in \mathcal{R}$ for some
  i-delayed faster-than relation~$\mathcal{R}$ and call $\dftpoi$ \emph{i-delayed faster-than preorder}.
\label{def:1-delayed-preorder}\qed
\end{defn}

As usual, $\dftpoi$ is the largest delayed faster-than relation; somewhat similar
to weak bisimulation, one can also show that it is a preorder.

It is shown in \cite{LV04} that $\nftpo$ and~$\dftpo$ coincide. To complete
the picture, we will show that $\dftpo$ and~$\dftpot$ coincide and get:

\begin{theorem}
  The preorders~$\nftpo$, $\nftpot$, $\dftpo$ and~$\dftpot$ coincide.
\label{thm:1-delayed-2-delayed-coincidence}
\end{theorem}

\begin{proof} We must only show the last coincidence; Case (2) of
\ref{def:1-delayed-preorder} is obvious. As above, we
consider some $P$ and $Q$ with $P \dftpo Q$.

\begin{enumerate}
\item[(1)] If $P \derives{\alpha} P'$, then 
$Q \derivesL{\sigma}1^{\ast}\!\derives{\alpha}\; \derivesL{\sigma}1^{\ast} Q'$ 
and hence 
$Q \derivesL{\sigma}2^{\ast}\!\derives{\alpha}\; \derivesL{\sigma}2^{\ast} Q'$
such that  $P' \dftpo Q'$ by definition of $\dftpo$. 

\item[(3)] If $P \derivesL{\sigma}2 P'$ for some $P'$, we have 
$P \derivesL{\sigma}1 P''$ for some $P''$ such that 
$P' \uptopotc P''$ by Proposition~\ref{lem:coherence}.
Thus, $Q \derivesL{\sigma}1^{+} Q''$ and, hence, $Q \derivesL{\sigma}2^{+} Q''$ 
for some $Q''$ with $P'' \dftpo Q''$ by the definition of $\dftpo$.
Further, $\uptopotc$ $\subseteq$ $\nftpo$ $=\dftpo$ by 
Proposition~\ref{prop:up-to-naives}(\ref{part:2-up-to-naives}) and \cite{LV04},
hence $P' \uptopotc P''$ implies $P' \dftpo P''$; we are done by
transitivity of $\dftpo$.
\end{enumerate}
 
For the reverse inclusion, we consider $(P,Q) \in \dftpot$.
For a smooth presentation, we start with the simulation of a time step.
 
\begin{enumerate}
 
\item[(3)] If $P \derivesL{\sigma}1 P'$, then $P \derivesL{\sigma}2 P'$ as well.
Then, we get some $Q'$ with $P' \dftpot Q'$ and $Q \derivesL{\sigma}2^{+} Q'$, i.e. 
$Q \derivesL{\sigma}2 Q'$ by transitivity. This gives us $Q \derivesL{\sigma}1 Q''$
with $Q' \uptopotc Q''$ by Proposition~\ref{lem:coherence}. Since 
$\nftpo$ and $\dftpo$ coincide and are included in $\dftpot$, we get
$Q' \dftpot Q''$ with 
Proposition~\ref{prop:up-to-naives}(\ref{part:2-up-to-naives}), and are
done by transitivity of $\dftpot$.


\item[(1)] For this part, cf.\ the figure below: if $P \derives{\alpha} P'$, then 
$Q \equiv Q_{0} \derivesL{\sigma}2 \dots \derivesL{\sigma}2 Q_{n} 
\derives{\alpha} Q_{n+1} 
\derivesL{\sigma}2 \dots \derivesL{\sigma}2 Q_{m} \equiv Q'$ for some $Q'$ 
and some $n$ and $m$ with $0 \le n < m$ such that $P'\;\dftpot\; Q'$ by 
definition of $\dftpot$.  
Considering that $Q_{0}\uptopotc Q_{0}$ by definition, we obtain
$Q_{0} \derivesL{\sigma}1 Q_{1}'\derivesL{\sigma}1 \dots \derivesL{\sigma}1 Q_{n}' 
\derives{\alpha} Q_{n+1}' 
\derivesL{\sigma}1 \dots \derivesL{\sigma}1 Q_{m}'$ for some $Q_{m}'$ 
with $Q_{m} \uptopotc Q_{m}'$ from 
Lemma~\ref{lem:uptopotc-sigma1} and
Proposition~\ref{prop:up-to-naives}(\ref{part:2-up-to-naives}).
As above, we derive $Q_{m}\;\dftpot\; Q_{m}'$ and are
done by transitivity of $\dftpot$ .\qed
\end{enumerate}
\end{proof}

\ifpdf
	\begin{figure}[htbp]
	\phantom{.}\hfill\includegraphics[clip=true]{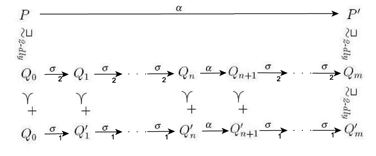}\hfill\phantom{.}
	\caption{}
	\label{fig-delay}
\end{figure}

\else
	\begin{figure}[htbp]
	\phantom{.}\hfill\includegraphics[scale=0.5, clip=true, angle=270]{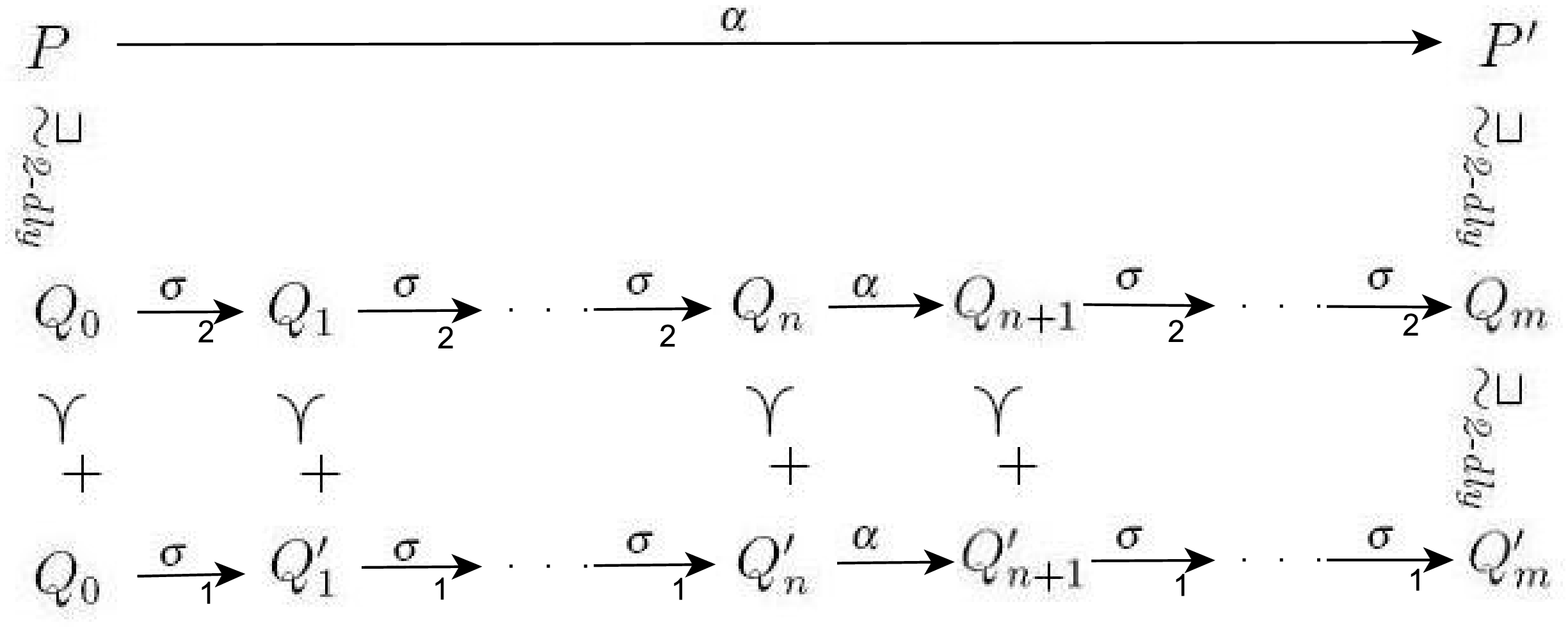}\hfill\phantom{.}
	\caption{}
	\label{fig-delay}
	\end{figure}
\fi

Clearly, any 1-naive faster-than relation is also a 1-delayed one, 
as well as any 2-naive faster-than relation a 2-delayed one.
But the sets of delayed relations of type 1, type 2 resp. are
incomparable, and this also holds for the naive relations. We come back
to this issue in Sec.~\ref{sec:strong-faster-than-precongruences}
in a slightly different setting.

\section{Indexed faster-than preorder}
\label{sec:indexed-faster-than-preorder}

The second variant of a faster-than preorder in \cite{LV04} is the 
indexed faster-than preorder, formalizing the idea of an 
account for time steps for the faster process.
If a time step of the slower process is not (or cannot be) simulated 
immediately by the faster process, 
then this time step is credited and might be withdrawn if the process 
performs this time step later on.
Obviously, the account balance may never be negative. This much more
complicated variant answers the question how time steps of the slower process
are matched.

\begin{defn}\rm For $i \in \{1,2\}$, a family $(\mathcal{R}_{j})_{j \in \nat}$ 
of relations over~$\proc$, 
  indexed by natural numbers (including~$0$), is a
  \emph{family of i-indexed faster-than relations} if,
  for all $j \in \nat$, $\pair{P}{Q}\in
  \mathcal{R}_{j}$ and $\alpha \in \act$:
  \begin{enumerate}
  \item\label{part:1-j-preorder} $P \derives{\alpha} P'$ implies
    $\exists Q'.\, 
    Q \derives{\alpha} Q'$ and $\pair{P'}{Q'}\in
    \mathcal{R}_{j}$.
  \item\label{part:2-j-preorder} $Q \derives{\alpha} Q'$ implies
    $\exists P'.\, P \derives{\alpha} P'$ and $\pair{P'}{Q'}\in
    \mathcal{R}_{j}$.
  \item\label{part:3-j-preorder} $P \derivesL{\sigma}i P'$ implies
    (a)~$\exists Q'.\, Q \derivesL{\sigma}i Q'$ and $\pair{P'}{Q'}\in
    \mathcal{R}_{j}$, or\\
    \phantom{$P \derivesL{\sigma}i P'$ implies}
    (b)~$j>0$ and $\pair{P'}{Q} \in \mathcal{R}_{j-1}$.
  \item\label{part:4-i-preorder} $Q \derivesL{\sigma}i Q'$ implies
    (a)~$\exists P'.\, P \derivesL{\sigma}i P'$ and $\pair{P'}{Q'}\in
    \mathcal{R}_{j}$, or\\
    \phantom{$Q \derivesL{\sigma}i Q'$ implies}
    (b)~$\pair{P}{Q'} \in \mathcal{R}_{j+1}$.
  \end{enumerate}
  We write $P \,\ftpo_{i,j}\, Q$ if $\pair{P}{Q}\in \mathcal{R}_{j}$ for
  some family of i-indexed faster-than relations~$(\mathcal{R}_{j})_{j
    \in \nat}$ and call $\ftpo_{i,j}$ \emph{i,j-indexed faster-than preorder}.
\label{def:indexed-preorder}\qed
\end{defn}

The latter notion is an abuse, since these relations are not really preorders.
$\cite{LV04}$ proves that $\ftpo_{1,0}$ coincides with the 1-naive faster-than 
preorder. We would have expected to be able to 
prove the same coincidence result in the new setting. 
Unfortunately, this has turned out to be wrong, which can be explained 
by the absence of time determinism. 
According to their definitions, $\ftpo_{2,0} \subseteq \nftpot$ obviously
holds.
However, the reverse inclusion fails; consider the following counterexample:
\\
Let $P \df \tau.\nil\,|\,\sigma.\sigma.\tau.\nil$ and 
$Q \df \sigma.\tau.\nil\,|\,\sigma.\sigma.\tau.\nil$. 
Clearly, $P$ is faster than $Q$ in the sense of a naive faster-than preorder. 
In the sequel, we will try to build a family of relations 
$(\mathcal{R}_{j})_{j \in \nat}$ such that 
$(P,Q) \in \ftpo_{2,0}$ holds.
Hence, we put $(P,Q)$ into $\mathcal{R}_{0}$ and first consider the 'real' 
type-2 time step $Q \derivesL{\sigma}2 \tau.\nil\,|\,\tau.\nil$; since
$P$ is not able to match this behaviour, as a time step is preempted by an 
urgent $\tau$, we are forced to credit this time step and hence obtain:
\\

$R_{0} \df \{(\tau.\nil\,|\,\sigma.\sigma.\tau.\nil,\, 
\sigma.\tau.\nil\,|\,\sigma.\sigma.\tau.\nil), \dots \}$

\begingroup\leftskip=2cm $R_{1} \df \{(\tau.\nil\,|\,\sigma.\sigma.\tau.\nil,\, 
\tau.\nil\,|\,\tau.\nil ), \dots \}$ \endgroup
\\

Now we consider $\tau.\nil\,|\,\sigma.\sigma.\tau.\nil \derives{\tau}
\nil |\,\sigma.\sigma.\tau.\nil$ of $P$, which '$Q$' can either mimic 
with \\
$\tau.\nil\,|\,\tau.\nil \derives{\tau} \nil |\,\tau.\nil$ or 
$\tau.\nil\,|\,\tau.\nil \derives{\tau} \tau.\nil |\,\nil$. 
Since the resulting processes have the same functional and waiting behaviour, 
we may consider any of them. Thus, we so far know:\\

$R_{0} \df \{(\tau.\nil\,|\,\sigma.\sigma.\tau.\nil,\, \sigma.\tau.\nil\,|\,\sigma.\sigma.\tau.\nil), \dots \}$

\begingroup\leftskip=2cm $R_{1} \df \{(\tau.\nil\,|\,\sigma.\sigma.\tau.\nil,\, \tau.\nil\,|\,\tau.\nil),(\nil\,|\,\sigma.\sigma.\tau.\nil,\, \nil\,|\,\tau.\nil), \dots \}$\endgroup
\\

Finally, $\nil\,|\,\sigma.\sigma.\tau.\nil$ may perform 
$\nil\,|\,\sigma.\sigma.\tau.\nil \derivesL{\sigma}2 \nil\,|\,\sigma.\tau.\nil$.
As this time step cannot be simulated by $Q$, we have to withdraw the credited 
time step and put $(\nil\,|\,\sigma.\tau.\nil,\, \nil\,|\,\tau.\nil)$ in the 
relation $\mathcal{R}_{0}$. This leads
to a contradiction. \\

Summarizing, the problem lies in the fact that the slower process $Q$ 
performs a 'real' type-2 time step and skips a $\sigma$-prefix, 
but only one time step is credited for the faster process $P$.
The amortisation mechanism does not ensure that $P$ can match this 
'real' type-2 time step properly later; it instead wastes its credit on
a type-1 time step. This cannot happen in a sensible setting with time 
determinism.
We leave the repair of this defect by altering the definition of 
the indexed faster-than 
preorder for future work.
 

\section{Strong faster-than precongruence}
\label{sec:strong-faster-than-precongruences}

A shortcoming of the 1-naive faster-than preorder is that it is not compositional
wrt. $|$. 
As an example consider the processes $P \df \sigma.a.0$ and 
$Q \df a.0$ for which $P \nftpoi Q$ holds: the time step 
$\sigma.a.0 \derivesL{\sigma}i a.0$ of $P$
is matched by the time step $a.0 \derivesL{\sigma}i a.0$, i.e. $Q$ idles. 
Yet, if we compose both processes in parallel with 
$R \df \overline{a}.\nil$, then we observe
that $P\,|\,R \;\nftpoi \; Q\,|\,R$ does not hold as the clock transition 
$P\,|\,R \derivesL{\sigma}i a.\nil|\overline{a}.\nil$
cannot be matched  due to an urgent $\tau$.

In any case, it is intuitively  suggestive to exclude pairs   
$\pair{\sigma.P}{P}$ from a faster-than preorder.
In \cite{LV04}, the 1-naive faster-than preorder is modified taking urgent sets
into account: if $Q$ performs a time step and $a\in \urgent Q$, then a context
process does not really have to wait for synchronisation on $a$, $Q$ ist
fast on $a$; thus,
when matching a time step the faster process must have a larger urgent set.
The resulting preorder turns out to be the largest precongruence contained
in $\nftpo$; it is called strong faster-than precongruence. Again,
we adopt the definition in our setting.

\begin{defn}\rm
  For $i \in \{1,2\}$,
  a relation $\mathcal{R}\subseteq \proc\times\proc$ is a \emph{strong
    i-faster-than relation} if the following conditions hold for all
  $\pair{P}{Q}\in \mathcal{R}$ and $\alpha \in \act$.
  \begin{enumerate}
  \item $P \derives{\alpha} P'$ implies $\exists Q'.\, Q
    \derives{\alpha} Q'$ and $\pair{P'}{Q'}\in \mathcal{R}$.
  \item $Q \derives{\alpha} Q'$ implies $\exists P'.\, P
    \derives{\alpha} P'$ and $\pair{P'}{Q'}\in \mathcal{R}$.
  \item $P \derivesL{\sigma}i P'$ implies $\urgent{Q} \subseteq
    \urgent{P}$ and $\exists Q'.\, Q \derivesL{\sigma}i Q'$ and
    $\pair{P'}{Q'}\in \mathcal{R}$.
  \end{enumerate}
  We write $P \,\ftpoi\, Q$ if $\pair{P}{Q}\in \mathcal{R}$ for some
  strong i-faster-than relation~$\mathcal{R}$ and call $\ftpoi$ \emph{strong i-faster-than precongruence}.
\label{def:strong-precongruence}\qed
\end{defn}

Clearly, $\ftpoi$ is contained in $\nftpoi$ for $i \in \{1,2\}$.
As usual, it is easy to prove that $\ftpoi$ is the largest strong i-faster-than 
relation and a preorder.  Since the $\nftpoi$ coincide, they contain the same 
largest precongruence. Hence, the expected coincidence result also shows 
the most important fact that $\twoftpo$ is this largest precongruence.
We first present the following result.
\begin{prop}
  The relation~$\uptopo$ satisfies the
  defining clauses of a strong 1-faster-than relation, also on open terms;
  hence, $\uptopo$ restricted to processes is a strong 1-faster-than relation and
  $\uptopo_{| \proc \times \proc} \;\,\df\,\;
  \uptopo \cap\, (\proc \times \proc)
  \;\subseteq \oneftpo$. The same holds for~$\uptopotc$.
\label{prop:up-to-precongruence}  
\end{prop}  

The first statement is given in \cite{LV04}; the second can be shown as we have
done above for
Proposition~\ref{prop:up-to-naives}.
Now the following coincidence result can be proven as indicated in the proof
of Theorem~\ref{thm:1-naive-2-naive-coincidence}. 

\begin{theorem}
  The preorders~$\oneftpo$ and~$\twoftpo$ coincide.
\label{thm:1-2-strong-faster-coincidence}
\end{theorem}

Although $\oneftpo$ and~$\twoftpo$ coincide, this does not hold for the
strong 1- and 2-faster-than relations. In fact, in some cases one can find a
smaller 2-faster-than relation because $Q$ has more possibilities to match
a clock transition. In other cases, one can find a smaller 
1-faster-than relation because $P$ can reach fewer processes by type-1
time steps; see \cite{Ilt09} for examples.

\bigskip

Another contribution of this paper is that we show how to combine these two advantages in
a mixed setting. The new \emph{strong combined faster-than} or 
\emph{strong c-faster-than precongruence} matches type-1 time steps with 
type-2 time steps. (The same idea is also studied on the naive level in
\cite{Ilt09}.)

\begin{defn}\rm
  A relation $\mathcal{R}\subseteq \proc\times\proc$ is a \emph{strong
    c-faster-than relation} if the following conditions hold for all
  $\pair{P}{Q}\in \mathcal{R}$ and $\alpha \in \act$.
  \begin{enumerate}
  \item $P \derives{\alpha} P'$ implies $\exists Q'.\, Q
    \derives{\alpha} Q'$ and $\pair{P'}{Q'}\in \mathcal{R}$.
  \item $Q \derives{\alpha} Q'$ implies $\exists P'.\, P
    \derives{\alpha} P'$ and $\pair{P'}{Q'}\in \mathcal{R}$.
  \item $P \derivesL{\sigma}1 P'$ implies $\urgent{Q} \subseteq
    \urgent{P}$ and $\exists Q'.\, Q \derivesL{\sigma}2 Q'$ and
    $\pair{P'}{Q'}\in \mathcal{R}$.
  \end{enumerate}
  We write $P \,\cftpo\, Q$ if $\pair{P}{Q}\in \mathcal{R}$ for some
  strong c-faster-than relation~$\mathcal{R}$ and call $\cftpo$ \emph{strong c-faster-than precongruence}.\qed
\label{def:c-strong-precongruence}
\end{defn}

As usual, $\cftpo$ is the largest strong c-faster-than relation.
However, note that it is not clear that $\cftpo$ is transitive.
This will follow with our next theorem. 

\begin{theorem}
  The preorders~$\oneftpo$ and~$\cftpo$ coincide.
\label{thm:1-c-strong-faster-coincidence}
\end{theorem}

\begin{proof}
First, take processes $P$ and $Q$ such that $P \oneftpo Q$; we only
have to consider the matching of time steps.
If $P \derivesL{\sigma}1 P'$ for some process $P'$, then 
$\urgent{Q} \subseteq \urgent{P}$ and $Q \derivesL{\sigma}1 Q'$ for some $Q'$ 
satisfying $P' \oneftpo  Q'$ by definition of $\oneftpo$.
We are done since we also have $Q \derivesL{\sigma}2 Q'$. 

\smallskip

For the reverse inclusion $\cftpo \;\subseteq \;\oneftpo$, define the relation 
$\mathcal{R}$ by $(P,Q) \in \mathcal{R}$ if and only if 
$\exists R \in \proc.\;P \cftpo R \uptopotc Q$ for $P, Q \in \proc$.
This relation contains $\cftpo$ since $\uptopotc$ is reflexive; hence, it
suffices to check that $\mathcal{R}$ is a strong 1-faster-than relation;
consider $P \cftpo R \uptopotc Q$.

\begin{enumerate}
\item [1.] If $P \derives{\alpha} P'$ for some $P'$, the definition of $\cftpo$ 
shows $R \derives{\alpha} R'$ for some process $R'$ with $P' \cftpo R'$. 
Since $\uptopotc$ is a strong 1-faster-than relation by 
Proposition~\ref{prop:up-to-precongruence}, this implies
$Q \derives{\alpha} Q'$ for some $Q'$ such that $R' \uptopotc Q'$.

\item[2.] The case $Q \derives{\alpha} Q'$ for some $Q'$ is analogous to Part (1).

\item[3.] If $P \derivesL{\sigma}1 P'$ for some $P'$, then 
$\urgent{R} \subseteq \urgent{P}$ and $R \derivesL{\sigma}2 R'$ for some process 
$R'$ with $P' \cftpo R'$ by definition of $\cftpo$. Due to 
$R \derivesL{\sigma}2 R'$, we infer $Q \derivesL{\sigma}1 Q'$ for some $Q'$
satisfying $R' \uptopotc Q'$ from Lemma \ref{lem:uptopotc-sigma1}.
Moreover, we know $\urgent{Q} \subseteq \urgent{R}$ due to $R \uptopotc Q$ by 
successive application of Lemma~\ref{lem:upto-aux1-rec}.5, implying
$\urgent{Q} \subseteq \urgent{P}$. \qed
\end{enumerate}
\end{proof}

The next result shows that strong c-faster-than relations indeed provide 
more possibilities to prove the strong faster-than precongruence.

\begin{prop}\text{}
\begin{enumerate}
   \item\label{part:1-coherence-relations}
Every strong 1-faster-than relation $\mathcal{R}$ is a strong c-faster-than relation. 
   \item\label{part:2-coherence-relations}
Every strong 2-faster-than-relation $\mathcal{R}$ is a strong c-faster-than relation.     
\end{enumerate}
\label{lem:coherence-relations}
\end{prop}

\begin{proof}
Again, we only have to look at the matching of time steps. First, take a 
strong 1-faster-than relation $\mathcal{R}$ and some  $(P,Q) \in \mathcal{R}$.
If $P \derivesL{\sigma}1 P'$, then $\urgent{Q} \subseteq \urgent{P}$ and
$Q \derivesL{\sigma}1 Q'$ for some process $Q'$ with $(P',Q') \in \mathcal{R}$ 
by definition of $\mathcal{R}$; then, also $Q \derivesL{\sigma}2 Q'$.

Now let $\mathcal{R}$ be a strong 2-faster-than relation and 
$(P,Q) \in \mathcal{R}$. If $P \derivesL{\sigma}1 P'$ for some process $P'$, 
then also $P \derivesL{\sigma}2 P'$. Hence, we obtain 
$\urgent{Q} \subseteq \urgent{P}$ and $Q \derivesL{\sigma}2 Q'$ for some 
$Q'$ such that $(P',Q') \in \mathcal{R}$ by the definition of $\mathcal{R}$.
\qed 
\end{proof}

We remark that one could also look at relations where a type-2 time step is 
matched by a type-1 time step; these lead to the same strong faster-than 
precongruence and comprise all relations that are strong 1- \emph{and} 2-faster-than
relations.

We conclude by giving processes $P$ and $Q$, where the least strong c-faster-than 
relation proving $P\cftpo Q$ is much smaller than any relation of type 1  or 2.
In the following, we write $\sigma^{n}$ for a sequence of $n$ $\sigma$-prefixes.
Let $P \df (\mu x.a.x) \;|\; (\mu x.b.x) \;|\; ((\mu x.\sigma^{n}.c.x) \;|\; (\mu x.\sigma^{n}.\overline{c}.x)) \resset{c}$ and  
$Q \df (\mu x.\sigma^{n}.a.x) \;|\; (\mu x.\sigma^{n}.b.x) \;|\; ((\mu x.\sigma^{n}.c.x) \;|\; (\mu x.\sigma^{n}.\overline{c}.x)) \resset{c}$ for some large $n$.
For a smooth presentation, we code $Q$ and its resulting processes as a subset of 
$\{ijkl\;|\;0 \le i,j,k,l \le n\}$, where $i$, $j$, $k$ and $l$ are the numbers of leading $\sigma$-prefixes for the four components.
Analogously, we code the $P$-states as elements of $\{ijkl\;|\; i,j\in\{0,n\},\;0 \le k,l \le n\}$ where $k$ and $l$ refer to the third and fourth component; $i=n$ encodes $\mu x.a.x$ (the original component as in all other cases) and $i=0$ encodes $a.\mu x.a.x$, resulting from a time step; similarly for $j$.

We first show that the following is a strong c-faster-than relation: 
$\{(ijkk,ijkk)\;|\;i,j \in \{0,n\}, 0 \le k \le n\}$, hence we obviously have $4(n+1)$ pairs.
Observe that, in each pair, the urgent set of the second process is contained in that of the first process, since $a$ and $b$ are always urgent on the $P$-side and $k$ is equal on both sides. 
Each time step on the $P$-side is of the form $ijkk \derivesL{\sigma}1 00 (k-1)(k-1)$ and is matched by the 'same'
time step of type-2 on the $Q$-side.
For the actions, observe that $a, b$ and $\tau$ set the resp. $i, j$ or $k$ to $n$ on both sides.

Next we show that a strong 2-faster-than relation for $(P,Q)$ must contain pairs
for $4n^{2} + 4$ different $P$-states: $P$ is encoded as $nnnn$, and there is 
a $\derivesL{\sigma}2$ transition to each $00kl$ with $0 \le k, l \le n-1$.
For each of these $n^{2}$ processes we have three additional ones, since action $a$ and/or action $b$ sets the first and/or second component to $n$. Furthermore, $0000$ can perform
$\tau$ to become $00nn$, and also for this process we have three additional ones.
Since $P$ can reach at least $4n^{2} + 4$ processes, a strong 2-faster-than relation must 
contain at least $4n^{2} + 4$ pairs.
Further factors $n$ are possible by replacing the $c$- and $\overline{c}$-components by more components that communicate in a ring:
e.g. $((\mu x.\sigma^{n}.c_{1}.\overline{c_{4}}.x)|(\mu x.\sigma^{n}.\overline{c_1}.c_{2}.x)|$ $(\mu x.\sigma^{n}.\overline{c_{2}}.c_{3}.x)
|(\mu x.\sigma^{n}.\overline{c_3}.c_{4}.x)) \resset{c1, \dots c4}$.

To present elements of a strong 1-faster-than relation for $(P,Q)$, we introduce the
following notation: we let $\lfloor i \rfloor_{n}$  be $n$ if $i=n$, and $0$ otherwise; 
furthermore, we write ${\it pair}(ijk)$ to denote the pair $\pair{\lfloor i \rfloor_{n}\lfloor j \rfloor_{n}\,kk }{ijkk}$ for $0 \le i, j, k \le n$.
We show now that the strong 1-faster-than relation must include all these pairs, hence we have at least $(n+1)^3$ pairs. 
By way of example, assume that $j=i-5$, $k=i-2$. Starting from ${\it pair}(nnn)$, three time steps give ${\it pair}((n-3)(n-3)(n-3))$, 
$\tau$ plus two time steps lead to ${\it pair}((n-5)(n-5)(n-2))$ and action $a$ gives ${\it pair}(n(n-5)(n-2))$. 
Now we can reach the desired pair by $n-i$ time steps.\\
Further factors $n+1$ can be obtained by adding further components $\mu x.a'.x$ to $P$ and $\mu x.\sigma^{n}.a'.x$ to $Q$; 
with each component, the factor 4 in the size of the c-relation only grows by a factor 2.

\section{Conclusion and future work}

In~\cite{LV04}, the process algebra TACS and three types of faster-than relations were 
introduced, and it was shown that the three types lead to the same faster-than preorder.
Here, we have extended the clock transitions of TACS processes by new 
time steps and studied, in this new setting, the three types of relations. 
With the exception of the indexed faster-than preorder, we have 
proved that the new definitions lead to the same preorder as in~\cite{LV04}.
We have also obtained a coincidence result for the strong faster-than precongruence 
of~\cite{LV04}, and developed a new type of faster-than relation that combines old and
new operational semantics; this combined variant leads to the same precongruence and allows smaller relations for proving that the precongruence holds.

In~\cite{Ilt09}, it is also shown that (at least on the naive level) faster-than
relations up to $\nftpoi$ can be employed, which can also lead to smaller relations
as usual. Furthermore, coincidence also holds for weak variants of the faster-than
preorder in the old and the new operational semantics.

Summarizing, we have given a number of new results that further support the
robustness of the approach in~\cite{LV04} for comparing the worst-case 
efficiency of asynchronous processes.

\smallskip
Although in \cite{Lyn96} hundreds of pages are devoted to a setting with upper time bounds only, this is rarely 
treated in process algebra. Typically, e.g. in \cite{HenReg95}, $\sigma$ represents a definite time step; 
then replacing a component by one with fewer time steps might change the overall system behaviour drastically instead
of improving efficiency.
The efficiency preorder studied in \cite{AruHen92} counts $\tau$'s; since parallel $\tau$'s count twice, while parallel $\sigma$'s are 
counted only once, this does not compare time but some other cost, as is done in \cite{AKK05}.
There is one other process algebra with upper time bounds, which has been studied in a testing scenario and related to fairness, see e.g. \cite{CorBerVog09, CorVog07}.

\smallskip

As future work, we intend to repair the defect of the indexed-faster-than relations in 
our new setting by altering its definition.
This approach is relevant since the indexed-faster-than preorder is quite a 
convincing candidate for a faster-than relation on processes. 
Another open issue is to consider combined definitions of the delayed and the weak variants. Moreover, one should study to what extent such a combined weak precongruence can be based on small relations.
Finally, to reinforce robustness of the approach further, there is at least one additional
intuitive variation of the operational semantics we will consider.

\nocite{*}

\bibliographystyle{eptcs}
\bibliography{robust}

\appendix


\end{document}